\documentclass{article}

\usepackage[letterspace=-100]{microtype}
\usepackage{graphicx}
\usepackage{subfigure}
\usepackage{booktabs} 
\usepackage{multirow}

\usepackage{acronym}

\newacro{IR}{Intermediate Representation}
\newacro{KL}{Kullback–Leibler}

\usepackage[utf8]{inputenc}
\usepackage{hyperref}
\usepackage{xcolor}
\usepackage{tabularx}
\usepackage{float}

\usepackage{tikz}

\usepackage{xspace}

\usepackage{amsmath,amsfonts,bm}

\def\eqref#1{equation~\ref{#1}}

\def\1{\bm{1}}

\def\rr{{\textnormal{r}}}

\def\va{{\bm{a}}}

\def\vh{{\bm{h}}}

\def\vr{{\bm{r}}}

\def\vv{{\bm{v}}}
\def\vw{{\bm{w}}}
\def\vx{{\bm{x}}}

\DeclareMathAlphabet{\mathsfit}{\encodingdefault}{\sfdefault}{m}{sl}
\SetMathAlphabet{\mathsfit}{bold}{\encodingdefault}{\sfdefault}{bx}{n}

\def\gE{{\mathcal{E}}}

\def\gL{{\mathcal{L}}}

\newcommand{\R}{\mathbb{R}}

\usepackage[arxiv]{icml2023}

\iffalse
\newcommand{\yuandong}[1]{\textcolor{red}{[Yuandong: #1]}}
\newcommand{\benoit}[1]{\textcolor{purple}{[Benoit: #1]}}
\newcommand{\chris}[1]{\textcolor{purple}{[chris: #1]}}
\newcommand{\todo}[1]{\textcolor{red}{ \underline{\textbf{[TODO]}}: #1}}
\newcommand{\kevin}[1]{\textcolor{purple}{[kevin: #1]}}
\newcommand{\youwei}[1]{\textcolor{purple}{[youwei: #1]}}
\else
\newcommand{\yuandong}[1]{}
\newcommand{\benoit}[1]{}
\newcommand{\chris}[1]{}
\newcommand{\todo}[1]{}
\newcommand{\kevin}[1]{}
\newcommand{\youwei}[1]{}
\fi

\usepackage{amsmath}
\usepackage{amssymb}
\usepackage{mathtools}
\usepackage{amsthm}

\usepackage[capitalize,noabbrev]{cleveref}

\theoremstyle{plain}
\newtheorem{theorem}{Theorem}[section]

\newtheorem{lemma}[theorem]{Lemma}

\theoremstyle{definition}

\theoremstyle{remark}

\icmltitlerunning{Learning Compiler Pass Orders using Coreset and Normalized Value Prediction}

\begin{document}

\twocolumn[
\icmltitle{Learning Compiler Pass Orders using Coreset and Normalized Value Prediction}


\begin{icmlauthorlist}
\icmlauthor{Youwei Liang$^*$}{uni}
\icmlauthor{Kevin Stone$^*$}{comp}
\icmlauthor{Ali Shameli}{comp}
\icmlauthor{Chris Cummins}{comp}
\icmlauthor{Mostafa Elhoushi}{comp}
\icmlauthor{Jiadong Guo}{comp}
\icmlauthor{Benoit Steiner}{comp}
\icmlauthor{Xiaomeng Yang}{comp}
\icmlauthor{Pengtao Xie}{uni}
\icmlauthor{Hugh Leather}{comp}
\icmlauthor{Yuandong Tian}{comp}
\end{icmlauthorlist}

\icmlaffiliation{comp}{Meta AI}
\icmlaffiliation{uni}{University of California, San Diego}

\icmlcorrespondingauthor{Kevin Stone}{kevinleestone@meta.com}

\icmlkeywords{Machine Learning, ICML, Reinforcement Learning, Program Optimization, Graph Neural Networks}

\vskip 0.3in
]



\printAffiliationsAndNotice{\icmlEqualContribution} 

\def\ee{\mathbb{E}}
\def\vphi{\boldsymbol{\phi}}
\def\vw{\mathbf{w}}
\def\vx{\mathbf{x}}
\def\vh{\mathbf{h}}
\def\dd{\mathrm{d}}
\def\diag{\mathrm{diag}}

\newcommand{\programl}{ProGraML\xspace}

\begin{abstract}
Finding the optimal pass sequence of compilation can lead to a significant reduction in program size and/or improvement in program efficiency. Prior works on compilation pass ordering have two major drawbacks. They either require an excessive budget (in terms of compilation steps) at compile time or fail to generalize to unseen programs. In this paper, for code-size reduction tasks, we propose a novel pipeline to find program-dependent pass sequences within 45 compilation calls. It first identifies a \emph{coreset} of 50 pass sequences via greedy optimization of a submodular function, and then learns a policy with Graph Neural Network (GNN) to pick the optimal sequence by predicting the normalized values of the pass sequences in the coreset. Despite its simplicity, our pipeline outperforms the default \texttt{-Oz} flag by an average of 4.7\% over a large collection (4683) of unseen code repositories from diverse domains across 14 datasets. In comparison, previous approaches like reinforcement learning on the raw pass sequence space may take days to train due to sparse reward, and may not generalize well in held-out ones from different domains. Our results demonstrate that existing human-designed compiler flags can be improved with a simple yet effective technique that transforms the raw action space into a small one with denser rewards. 
\end{abstract}

\begin{figure}
    \centering
    \includegraphics[width=0.72\linewidth]{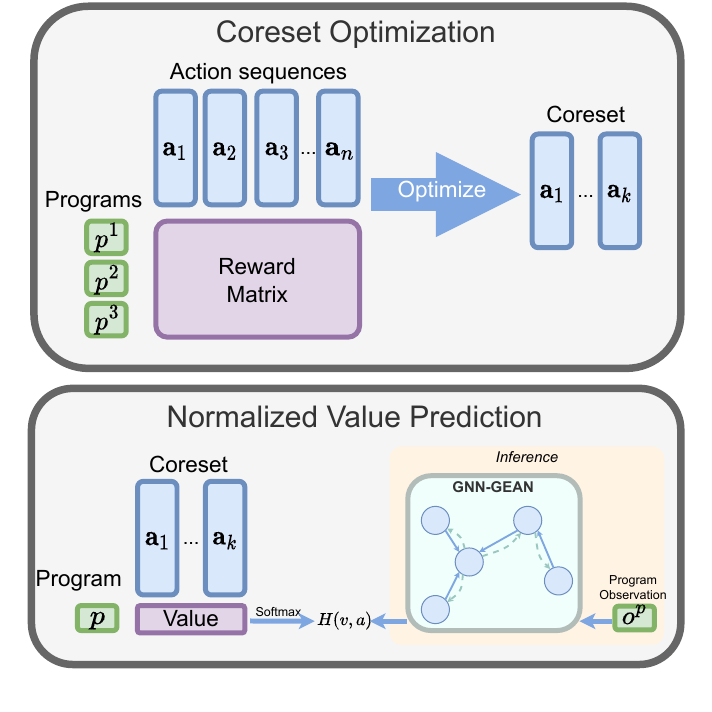}
    \caption{A depiction of our main contributions. (Top) \textbf{Coreset Optimization}: A process for discovering a small set of pass sequences (\textit{coreset}) that generalizes. (Bottom) \textbf{Normalized Value Prediction} A process where our model learns to predict the normalized value of pass sequences from the coreset. (Bottom inset) Our model, Graph Edge Attention Network (\textbf{GEAN}), for encoding program observations.}
    \label{fig:overview_figure}
\end{figure}

\section{Introduction}

For more efficient execution with fewer resources (e.g., memory, CPU, and storage), applying the right ordering for compiler optimization passes to a given program, i.e., \emph{pass ordering}, is an important yet challenging problem. Manual efforts require expert knowledge and are time-consuming, error-prone, and often yield sub-par results, due to the huge size of the search space. For example, the LLVM compiler has 124 different compilation flags. If the pass sequences have a length of 45, then the possible number of sequences ($124^{45} \sim 10^{94}$) is already more than the atoms in the universe ($\sim 10^{80}$~\cite{atoms}). 

In recent years, machine learning (ML)-guided pass ordering has emerged as an interesting field to replace this laborious process~\citep{ml4copt-survey}. Along this line, many works show promising results using various optimization and/or ML techniques (e.g., reinforcement learning~\citep{neurovectorizer}, language modelling~\cite{e2e-dl}, evolutionary algorithms~\citep{kulkarni2012mitigating}, etc). 

However, there are several limitations. Some previous works (e.g., MLGO~\cite{MLGO}, MLGoPerf~\cite{MLGOPerf}) run adaptive search algorithms to optimize a set of programs for many hours. While this achieves strong performance gain, the procedure can be slow and does not distill knowledge from past experience and requires searching from scratch for unseen programs. Recent works like Autophase~\cite{AutoPhase} learn a \emph{pass policy} via reinforcement learning, and applies it to unseen programs without further search procedure, and GO~\cite{zhou2020transferable} fine-tunes the models for unseen programs.  
These approaches work for unseen programs from the same/similar domain, but can still be quite slow in the training stage, and do not show good generalization to programs from very different domains. 

In this work, we propose a novel pass ordering optimization pipeline to reduce the \emph{code size} of a program. 
As a first contribution, we formulate the search for a \emph{universal} core set of pass sequences (termed \textbf{coreset}) as an optimization problem of a submodular reward function and use a greedy algorithm to approximately solve it. The resulting coreset consists of 50 pass sequences with a total number of passes of 625. Importantly, it leads to very strong performance across programs from diverse domains, ranging from the Linux Kernel to BLAS. 
Specifically, for one unseen program in the evaluation set, there exists one of the 50 pass sequences that leads to an average code size reduction of $5.8\%$ compared to the default \texttt{-Oz} setting, across 10 diverse codebases (e.g. Cbench~\cite{cbench}, MiBench~\cite{mibench}, NPB~\cite{npb}, CHStone~\cite{chstone}, and Anghabench~\cite{anghabench}) of over one million programs in total. Considering the huge search space of compiler flags, this is a very surprising finding. 

While it is time-consuming to find an optimal pass sequence with an exhaustive enumeration of the core subset of pass sequences, as a second contribution, we find that the (near) optimal pass sequence can be directly predicted with high accuracy via our \emph{Graph Edge Attention Network} (GEAN), a graph neural network (GNN) architecture adapted to encode the augmented ProGraML~\cite{deepdataflow} graphs of programs. 
Therefore, we can run a few pass sequences selected by the model on an unseen program to obtain a good code size reduction. This enables us to find a good flag configuration that leads to 4.7\% improvement on average, with just 45 compilation passes, 
a reasonable trade-off between the cost spent on trying compilation passes and the resulting performance gain. 

We compare our approach with extensive baselines, including reinforcement learning (RL) -based methods such as PPO, Q-learning, and behavior cloning. We find that RL-based approaches operating on the original compiler pass space often suffer from unstable training (due to inaccurate value estimation) and sparse reward space and fail to generalize to unseen programs at inference. 
In comparison, our approach transforms the vast action space into a smaller one with much more densely distributed rewards. In this transformed space, approaches as simple as behavior cloning can be effective and generalizable to unseen programs.

\section{Related Work}
Graph structured data are present in numerous applications and it has been shown that taking advantage of this data can help us train very effective machine learning models. \cite{cdfg} use abstract syntax trees and control flow graphs for learning compiler optimization goals. They show that using such graphs allows them to outperform state-of-the-art in the task of heterogeneous OpenCL mapping. \cite{graphcodebert} uses a transformer based model with a graph guided masked attention that incorporates the data flow graph into the training. They achieve state of the art performance in four tasks including code search, clone detection, code translation, and code refinement.

As a contender to graph neural networks, \cite{graphit} uses transformers to process graphs. They show that if we effectively encode positional and local sub-structures of graphs and feed them to the transformer, then the transformer can outperform the classical GNN models. They test their model on classification and regression tasks and achieve state of the art performance. \chris{Seems unrelated:} In~\cite{curl}, they used an unsupervised model to learn embeddings of high dimensional pixel data using contrastive learning. They then use this embedding for downstream reinforcement learning tasks.

\section{Methodology}

\subsection{Action space}

The CompilerGym framework~\cite{compilergym} provides a convenient interface for the compiler pass ordering problem. The default environment allows choosing one of $124$ discrete actions at each step corresponding to running a specific compiler pass. In this work we will use the term pass interchangeably with action. We fix episode lengths to 45 steps to match the setup in~\cite{AutoPhase}. Given that our trajectories have a length of 45 steps, this means we have $124^{45}\sim 1.6\times 10^{94}$ possible pass sequences to explore. To find an optimal pass sequence for a program, we can apply some existing reinforcement learning methods including Q learning like DQN~\cite{dqn} and policy gradient like PPO~\cite{ppo}. 

\textbf{Pass Sequences} However for this problem it turns out that certain pass sequences are good at optimizing many different programs (where ``good'' is defined as better than the compiler default \texttt{-Oz}). We found that constraining the action space to a learned set of pass sequences enables state of the art performance and also significantly reduces the challenge of exploration. This allows us to cast the problem as one of supervised learning over this set of pass sequences. We use the following algorithm to find a good set of pass sequences.

\begin{figure}[t!]
    \centering
    \includegraphics[width=0.6\linewidth]{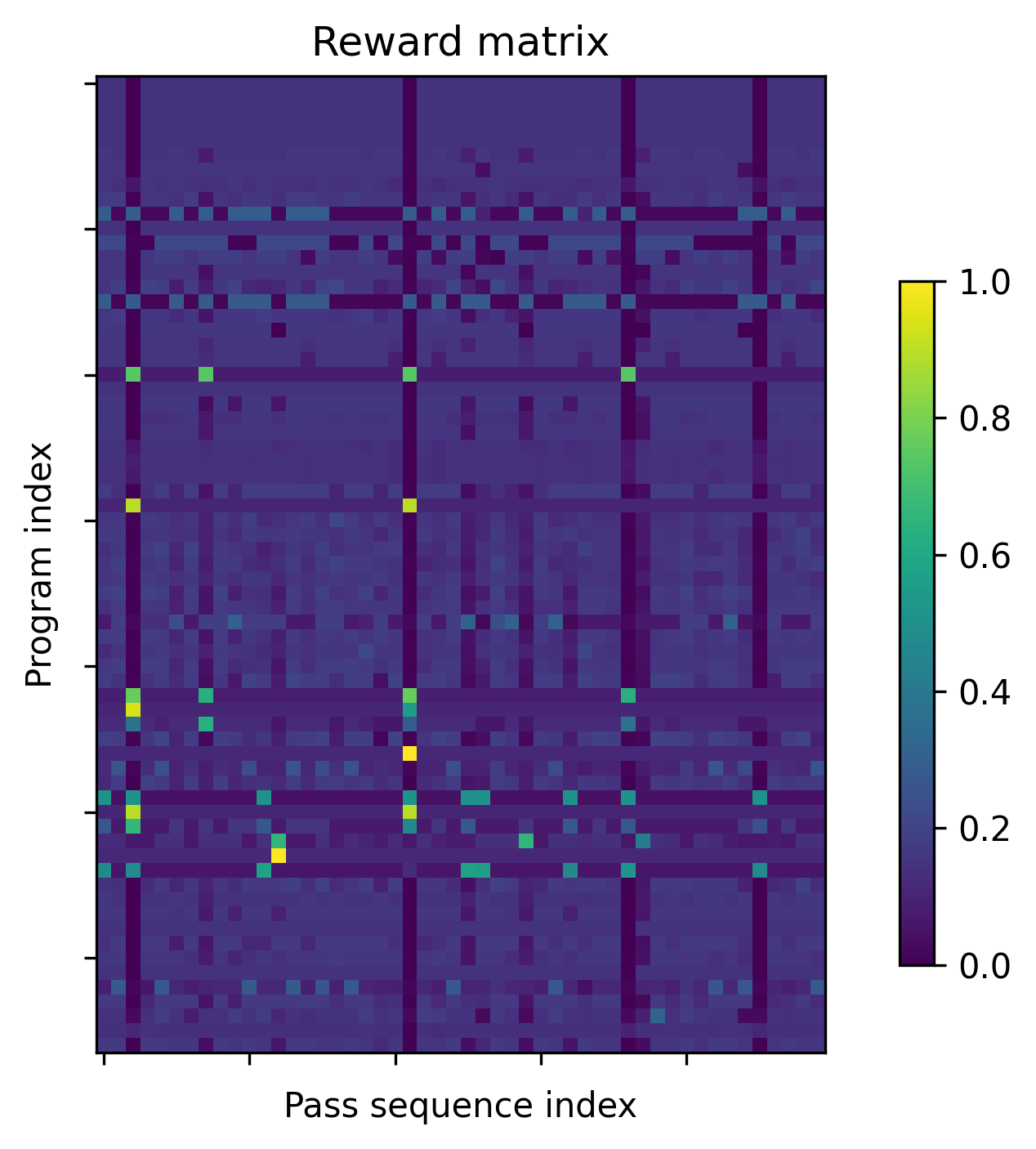}
    \caption{An exemplar reward matrix for 67 programs and 50 pass sequences. Most of the pass sequences do not lead to strong rewards, except for a few. On the other hand, certain pass sequences (i.e., columns) can lead to high rewards for multiple programs simultaneously and thus are good candidates for the coreset.}
    \label{fig:reward_matrix}
\end{figure}

\def\rr{\mathbb{R}}

Suppose we have $N$ programs and $M$ promising pass sequences. Let $R = [r_{ij}] \in \rr^{N\times M}$ be the reward matrix, in which $r_{ij} > 0$ is the \emph{ratio} of the codesize of $i$-th program if applied with $j$-th pass sequence, compared to \texttt{-O0}. $r_{ij} > 1$ means that the $j$-th pass sequence does \emph{better} than \texttt{-O0} in codesize reduction for $i$-th program, and $r_{ij} < 1$ means it performs worse. The reward matrix is normalized per row, by the maximum reward for each program, so that the optimal pass sequence has reward of 1 for each program. 

Then we aim to pick a subset $S$ of $K$ pass sequences, called the \emph{coreset}, from all $M$ pass sequences, so that the overall saving $J(S)$ is maximized:

\begin{equation}
\max_{|S| \le K} J(S) = \sum_{i=1}^N \max_{j \in S} r_{ij} \label{eq:obj-picking-coreset}
\end{equation}

\textbf{Finding $M$ candidate pass sequences}. Note that there can be exponential number of pass sequences, and we cannot construct the entire reward matrix, instead we seed a list of candidate action trajectories. For this, we run a random policy on a subset of $M$ (17500) selected training programs. For each selected program, we run $E$ (200) episodes and pick the best pass sequence as one of the $M$ candidates. If part of the best pass sequence leads to the same state, they are truncated so that the sequence becomes shorter. If multiple pass sequences yield the same reward we only retain the first after ordering them length-lexicographically. On average these last two steps reduce the length of the candidate pass sequences by $80\%$.

We then construct the reward $r_{ij}$ by applying the $j$-th best pass sequence to program $i$, and comparing it with \texttt{-O0}.

\textbf{Finding the best coreset $S$ with greedy algorithm}. As a function defined in subsets, $J(S)$ can be proven to be a nonnegative and monotone submodular function (See Appendix). While maximizing a submodular function is NP-hard~\cite{ward2016maximizing}, the following \emph{greedy algorithm} is proven to be an efficient approximate algorithm that leads to fairly good solutions~\cite{nemhauser1978analysis}. Starting from $S_0 = \emptyset$, at each iteration $t$, it picks a new pass sequence $j_t$ as follows:
\begin{equation}
j_t := \arg\max_{j \notin S_{t-1}} J(S_{t-1} \cup \{j\})
\end{equation}
And $S_t \leftarrow S_{t-1} \cup \{j_t\}$ until we pick $K$ pass sequences. We set $K = 50$ in this paper.

Given the discovered coreset $S$, we define a \textbf{generalized action} as a pass sequence in $S$. Applying a generalized action to a program means that we roll out the corresponding pass sequence on the program and return the best program state (i.e., having the highest cumulative reward) (it is feasible because we can cache the program state for each step).

\subsection{Normalized Value Prediction} \label{sec:nvp}

After discovering the ``good'' pass sequences (i.e., the coreset), we can turn the problem of the sequential decision-making on compiler passes into a problem of supervised learning. The target is to train a model to predict the best pass sequence conditioned on the program, where the label of the program is the index of the pass sequence that results in the largest code size reduction. However, one important observation we have is that there are typically multiple pass sequences in the coreset that all result in the largest code size reduction (see Figure~\ref{fig:gnn_qvalue} for the examples). Therefore, instead of using the single-class classification method with a cross entropy loss, we leverage the fact we have access to the values for all pass sequences. We predict the softmax normalized value of each pass sequence with a cross entropy loss detailed below. This approach is similar to behavior cloning~\cite{bc} but with soft targets over the coreset.

For a program $p$, we roll out all pass sequences in the coreset on it, obtaining a reward $r_j^p$ for the $j$-th sequence (i.e., the highest cumulative reward observed during the rollout of the pass sequence), which forms a value vector $\vr^p = [r_1^p, \dots, r_K^p]$. Then, the normalized values of the pass sequences are defined by
\begin{align}\label{eq:normalizing_values}
    \vv^p = \operatorname{Softmax}(\vr^p / T)
\end{align}
where $T$ is a temperature parameter.

For an initial observation $o^p$ of a program, our model outputs a probability distribution, $\va = f(o^p)$, over the pass sequences. The target of the training is to make $\va$ close to the normalized values of the pass sequences. To this end, we use the \ac{KL} divergence to supervise the model, which can be reduced to the following cross entropy loss up to a constant term.
\begin{align}
    \gL(\va^p, \vv^p) = -\sum_{j=1}^K a_j^p \log v_j^p
\end{align}

\subsection{Program Representations} 
Since we use the CompilerGym~\citep{compilergym} environments for program optimization, we exploit the program representations from CompilerGym, where program source code is converted to LLVM \ac{IR}~\citep{llvm} and several representations are constructed from the IR, including the \programl graph~\citep{deepdataflow}, the Autophase feature~\citep{AutoPhase}, and the Inst2vec feature~\citep{ben2018neural}. We choose to use the LLVM environment from CompilerGym because the LLVM ecosystem is a popular compiler infrastructure that powers Swift, Rust, Clang, and more. 

\textbf{Autophase} We use the Autophase features~\citet{AutoPhase} to build some baseline models for comparison, which will be detailed in Section~\ref{sec:baseline_methods}. The Autophase feature is a 56-dimension integer feature vector summarizing the LLVM IR representation, and it contains integer counts of various program properties such as the number of instructions and maximum loop depth. we use an MLP to encode it and output a program representation.

\textbf{ProGraML} As a part of our main method, we leverage \programl~\citep{deepdataflow} graphs for training GNN models. \programl is a graph-based representation that encodes semantic information of the program which includes control flow, data flow, and function call flow. This representation has the advantage that it is not a fixed size - it does oversimplify large programs - and yet it is still a more compact format than the original \ac{IR} format. 
We list three bullet points of the \programl graph below.
\begin{itemize}
    \item \textbf{Node features} Each node in a \programl graph have 4 features described in Table~\ref{tab:gnn_features}. 
    The ``text'' feature is a textual representation and the main feature that captures the semantics of a node. For example, it tells us what an ``instruction'' node does (e.g., it can be \texttt{alloca}, \texttt{store}, \texttt{add}, etc). 
    \item \textbf{Edges features} Each edge in a \programl graph have 2 features described in Table~\ref{tab:gnn_features}. 
    \item \textbf{Type graph} There is an issue with the \programl graph. Specifically, a node of type variable/constant node can end up with a long textual representation (for the ``text'' feature) if it is a composite data structure. For example, a struct (as in C/C++) containing dozens of data members needs to include all the members in its ``text'' feature. In other words, the current \programl representation does not automatically break down composite data types into their basic components. Since there is an unbounded number of possible structs, this prevents 100\% vocabulary coverage on any IR with structs (or other composite types). 
    To address this issue, we propose to expand the node representing a composite data type into a type graph. Specifically, a pointer node is expanded into this type graph: {\lsstyle\verb|[variable] <- [pointer] <- [pointed-type]|}, where {\lsstyle\verb|[...]|} denotes a node and {\lsstyle\verb|<-|} denotes an edge connection. A struct node is expanded into a type graph where all its members are represented by individual nodes (which may be further expanded into their components) and connected to a {\lsstyle\verb|struct|} node. An array is expanded into this type graph: {\lsstyle\verb|[variable] <- [array] <- [element-type]|}. The newly added nodes are categorized as \texttt{type} nodes and the edges connecting the type nodes are \texttt{type} edges. The type nodes and type edges constitute the type sub-graphs in the \programl graphs. In this manner, we break down the composite data structures into the type graphs that consist of only primitive data types such as \texttt{float} and \texttt{i32}. 
\end{itemize}

\subsection{Network Architecture}
\label{sec:gcn-architecture}

Since the Autophase feature can be encoded by a simple MLP, we discuss only the network architectures for encoding the \programl graphs in this section.

We use a graph neural network (GNN) as the backbone to encode the \programl graphs and output a graph-level representation. The GNN encodes the graph via multiple layers of message passing and outputs a graph-level representation by a global average pooling over the node features. The goal of graph encoding is to use the structure and relational dependencies of the graph to learn an embedding that allows us to learn a better policy. To this end, we experimented with several different GNN architectures such as Graph Convolutional Network (\textbf{GCN})~\citep{gcn}, Gated Graph Convolutions Network (\textbf{GGC})~\cite{li2015gated}, Graph Attention Network (\textbf{GAT})~\cite{gatv2}, Graph Isomorphism Network (\textbf{GIN})~\citep{gin}. To better capture the rich semantics of node/edge features in the \programl graphs, we propose Graph Edge Attention Network (\textbf{GEAN}), a variant of the graph attention network~\citep{velivckovic2017graph}. These GNNs leverage both the node and edge features, so we start by presenting how to embed the node and edge features.

\begin{table}[t!]
    \centering
    \resizebox{.96\columnwidth}{!}{%
    \begin{tabular}{lll}
    \toprule
    & \textbf{Feature} & \textbf{Description} \\
    \midrule
    \multirow{3}{*}{\rotatebox[origin=c]{90}{Node}}& type & One of \{instruction, variable, constant, \textbf{type}\} \\
        & text  & Semantics of the node \\
         & function & Integer position in function \\
         & block & Integer position in IR basic block \\
    \midrule
    \multirow{2}{*}{\rotatebox[origin=c]{90}{Edge}}&  flow  & Edge type. One of \{call, control, data, \textbf{type}\} \\
    & position & Integer edge position in flow branching \\
    \bottomrule
    \end{tabular}%
    }
    \caption{Features in the \programl graph representation which we augment with type information (changes highlighted). We ablate the augmentations in Section~\ref{ablation-studies}.}
    \label{tab:gnn_features}
\end{table}

\textbf{Node embedding} For the ``text'' features of the nodes, we build a vocabulary that maps from text to integer. The vocabulary covers all the text fields of the nodes in the graphs in the training set. The final vocabulary consists of 117 unique textual representations, and we add an additional item ``unknown'' to the vocabulary which denotes any text features that may be encountered at inference time and we have not seen before. The $i$-th textual representation is embedded using a learnable vector $\vv_i \in \R^d$, where $d$ is the embedding dimension. The ``type'' feature is not used because it can be inferred from the ``text'' feature.

\textbf{Edge embedding} The edge embedding is the sum of three types of embedding as the following.
\begin{itemize}
    \item \textbf{Type embedding} We have 4 types of edge flows, so we use 4 learnable vectors to represent them. 
    \item \textbf{Position embedding} The ``position'' feature of an edge is a non-negative integer which does not have an upper bound. We truncate any edge positions larger than 32 to 32 and use a set of 32 learnable vectors to represent the edge positions. 
    \item \textbf{Block embedding} We use the block indices of the two nodes connected by the edge to construct a new edge feature. The motivation is that whether the edge goes beyond an IR basic block can influence program optimization. Suppose the block indices of the source node and the target node of an edge are respectively $b_i$ and $b_j$. We get the relative position of the two nodes with respect to IR basic blocks in the following way: $p_{block} = \text{sign}(b_i - b_j)$. If the edge connects two nodes in the same IR basic block, then $p_{block}$ is 0. And $p_{block}=\pm 1$ indicates the edge goes from a block to the next/previous block. There are 3 possible values for $p_{block}$, so it is embedded using 3 learnable vectors.
\end{itemize}
The final embedding of an edge is the sum of its type, position, and block embedding vectors.

\textbf{Graph mixup} We note that the \programl graphs of two programs can be composed into a single graph without affecting the semantics of the two programs. And their value vectors can be added up to correctly represent the value vector of the composite graph. In this manner, we can enrich the input space to the GNNs and mitigate model overfitting for the normalized value prediction method. 

\textbf{Graph Edge Attention Network} We introduce the GEAN in this paragraph and defer its mathematical details to the Appendix. There are two main differences between the GAT and GEAN. 1) GEAN adopts a dynamic edge representation. Specifically, GAT uses the node-edge-node feature to calculate the attention for neighborhood aggregation, while GEAN uses the node-edge-node feature to calculate not only the attention but also a new edge representation. Then, the updated edge representation is sent to the next layer for computation. Note that GAT uses the same edge embedding in each layer. We conduct an ablation study showing that the edge representation in GEAN improves the generalization of the model. 
2) GAT treats the graph as an undirected graph while GEAN encodes the node-edge-node feature to output an updated node-edge-node feature, where the two updated node features represent the feature to be aggregated in the source node and the target node, respectively. This ensures that the directional information is preserved in the neighborhood aggregation.

\subsection{Dataset Preparation}

\def\cs{\texttt{Curated Suite}}
\def\bs{\texttt{Open Source}}
\def\sg{\texttt{Synthetically Generated}}

Overfitting issues could happen if training is performed on a small subset of programs, or the set of programs is not diverse enough. To mitigate this we find it helpful to create an aggregate dataset that uses many different public datasets as curated by CompilerGym. CompilerGym gives us access to $14$ different datasets constructed using two different methods.

\begin{itemize}
\item \texttt{Curated}
These are small collections of hand-picked programs. They are curated to be distinct from one another without overlap and are not useful for training. Typically programs are larger as they may comprise multiple source files combined into a single program. These are commonly used for evaluating compiler optimization improvements.
\item \texttt{Uncurated}
These are comprised of individual compiler IRs from building open source repositories such as Linux and Tensorflow. We also include synthetically generated programs, targeted for compiler testing (not optimization).
\end{itemize}

For our aggregate dataset we decided to holdout the entirety of the four curated datasets for use as an out-of-domain test set. This is important because they represent the types of programs we expect to see in the wild. We also split the uncurated datasets into train, validaton, and test programs.

\begin{table}[hbt!]
\centering
\resizebox{.9\columnwidth}{!}{%
\begin{tabular}{llrrr}
\toprule
Type & Dataset & Train & Val & Test \\
\midrule
\multirow{10}{*}{Uncurated}&anghabench-v1&70,7000&1,000&2,000 \\
&blas-v0&133&28&29 \\
&github-v0&7,000&1,000&1,000 \\
&linux-v0&4,906&1,000&1,000 \\
&opencv-v0&149&32&32 \\
&poj104-v1&7,000&1,000&1,000 \\
&tensorflow-v0&415&89&90 \\
&clgen-v0&697&149&150 \\
&csmith-v0&222&48&48 \\
&llvm-stress-v0&697&149&150 \\
\midrule 
\multirow{4}{*}{Curated}&cbench-v1&0&0&11 \\
&chstone-v0&0&0&12 \\
&mibench-v1&0&0&40 \\
&npb-v0&0&0&121 \\
\midrule
Total& - &728,219&4,495&4,683 \\

\bottomrule
\end{tabular}%
}
\caption{CompilerGym dataset types and training splits. The hand curated datasets are used solely to evaluate generalization to real world program domains at test time.}
\end{table}

\subsection{Evaluation}
For all our metrics and rewards we leverage the IR instruction count as value we are trying to minimize. We also report metrics on each CompilerGym dataset as well as the mean over datasets to get a single number to compare overall results.

\begin{itemize}
\item
The mean percent improved over \texttt{-Oz} (\textbf{MeanOverOz}) defined as following:
\begin{equation} \label{eq:meanOverOz}
  \Bar{I}^{Oz} = \mathrm{MeanOverOz} := \frac{1}{\mathcal{|P|}}
  \sum_p{\frac{I^{Oz}_p - I^{\pi_\theta}_p}{I^{Oz}_p}},
\end{equation}
where $p$ is a specific program from the set of programs $\mathcal{P}$ in the dataset. $I^{Oz}_p$ is the number of IR instructions in the program after running the default compiler pass \texttt{-Oz}. $I^{\pi_\theta}_p$ is the number of IR instructions in the program after applying the policy under consideration. We can think of this as a simple average of the percent improvement over \texttt{-Oz}. 
\item
We also look compare the geometric mean (\textbf{GMeanOverOz}) of final sizes across all programs relative to \texttt{-Oz} to give a weighted comparison that is insensitive to outliers.
\begin{equation}
\Bar{I}_{G}^{Oz} = \mathrm{GMeanOverOz} := \left( \displaystyle  \prod_p{\frac{I^{Oz}_p}{I^{\pi_\theta}_p}}  \right)^{\frac{1}{\mathcal{|P|}}}
\end{equation}

\end{itemize}

\section{Experiments}

\subsection{Experimental Setup}

For the methods in Table~\ref{table:evalSummary}, we search over a set of hyper-parameters, including the temperature $T$ in Eq.~\ref{eq:normalizing_values}, number of layers, the embedding dimension of the node/edge features, and the output dimension in the hidden layers in the MLPs. We select the best model of each method based on the validation metric (validation MeanOverOz in 45 steps) and report the best model's metrics on the test set.

\begin{figure*}[t]
    \centering
    \begin{tabular}{ccc}
    \includegraphics[width=0.3\linewidth]{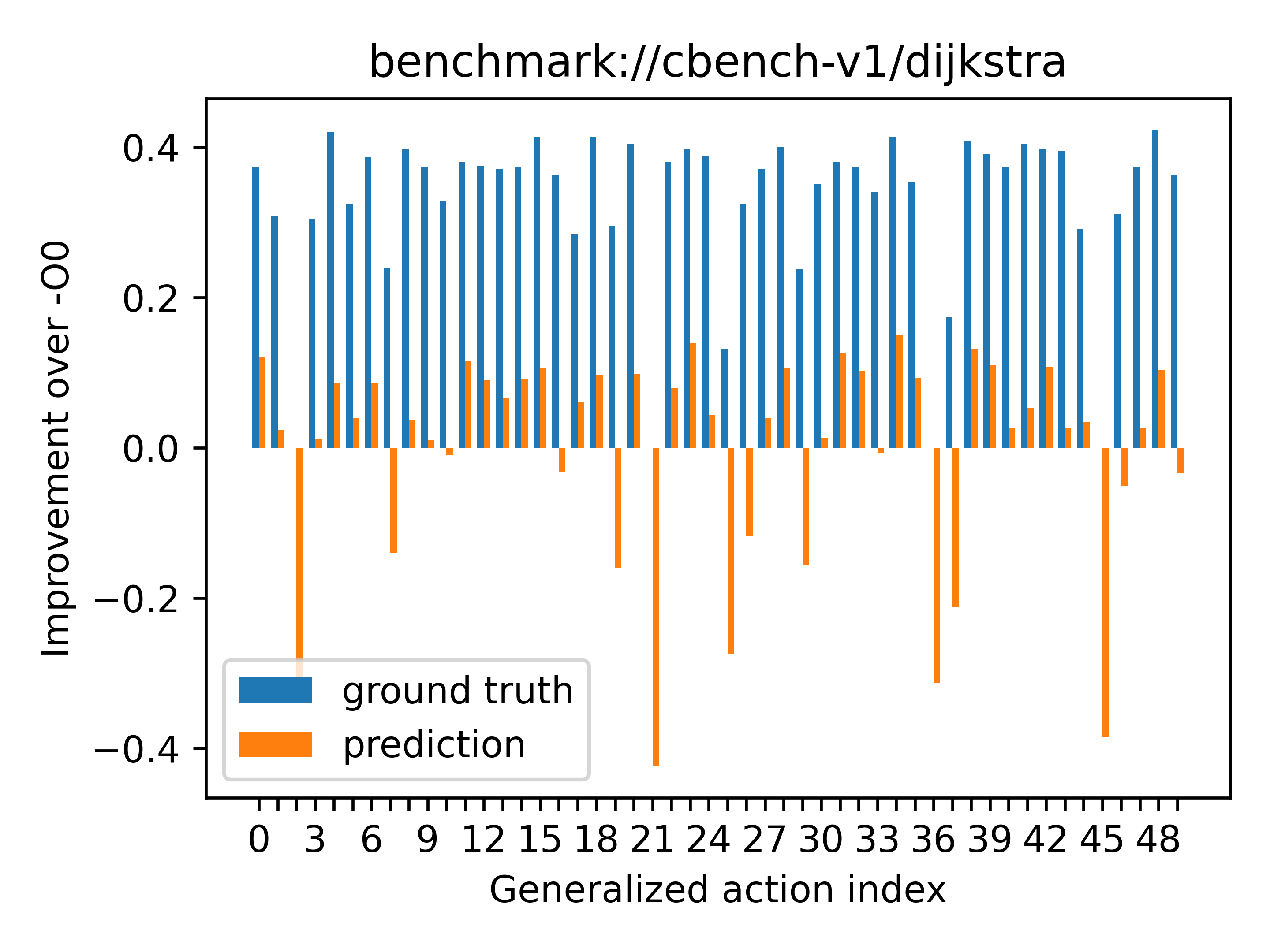} &
    \includegraphics[width=0.3\linewidth]{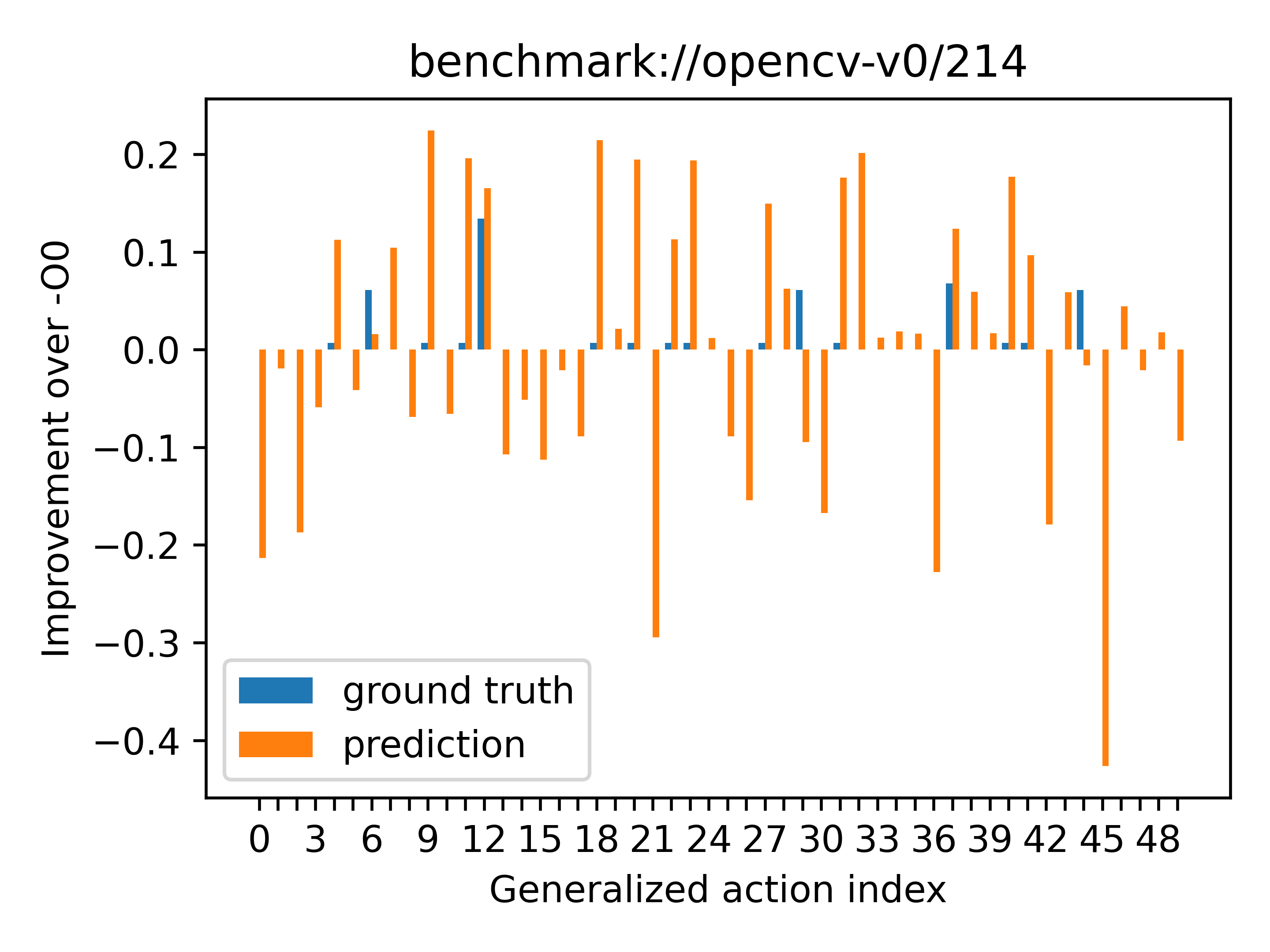} &
    \includegraphics[width=0.3\linewidth]{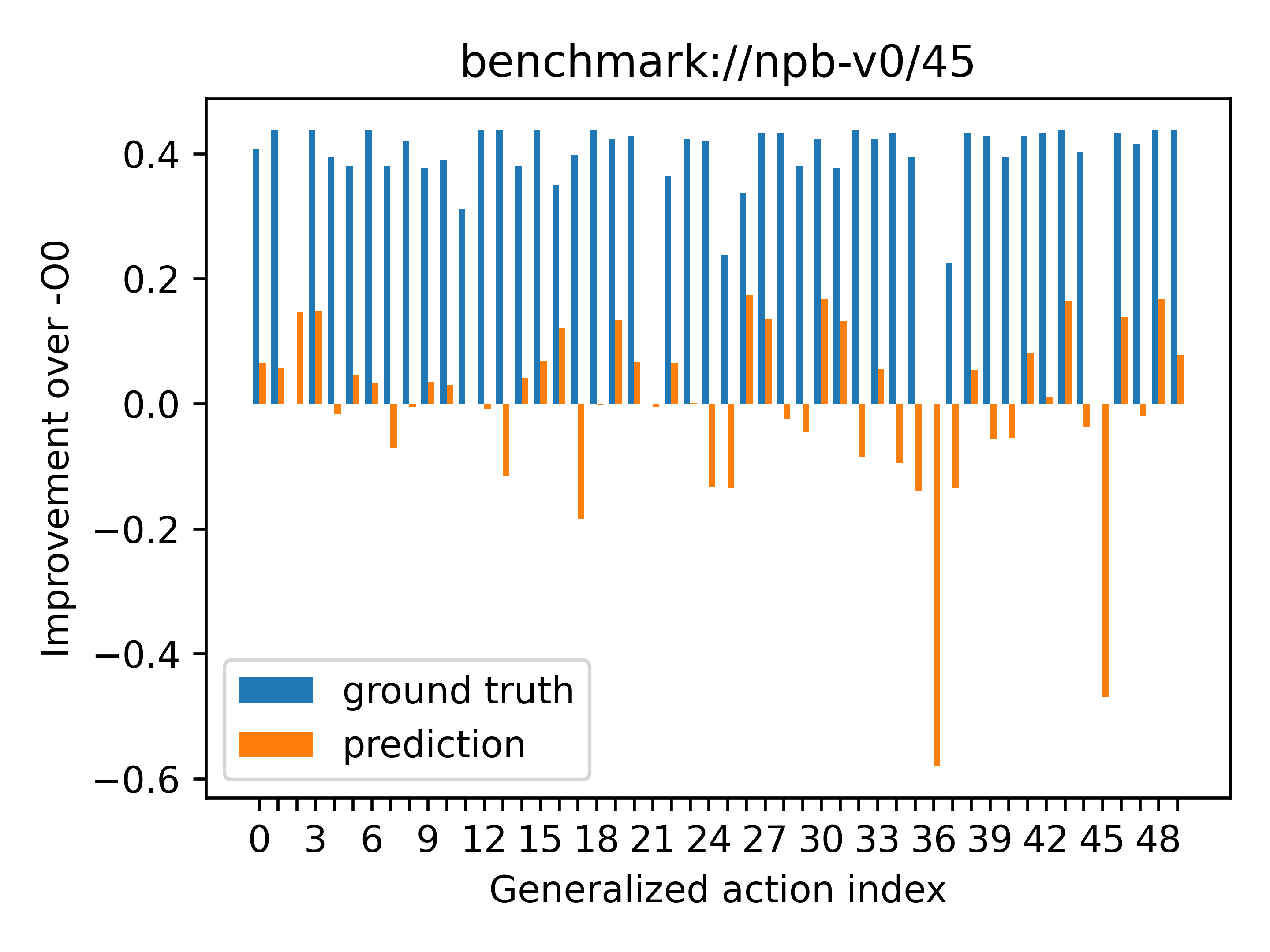}
    \end{tabular}
    \caption{\texttt{GEAN-Q-value-rank}: ground truth of rewards and model predictions over the 50 generalized actions for three benchmarks.}
    \label{fig:gnn_qvalue}
\end{figure*}

\subsection{Baseline Methods}
\label{sec:baseline_methods}
\textbf{Oracle} We consider a brute-force search over the coreset in order to find the best pass sequence for a given program. This gives us an upper-bound of the downstream policy network. In our case the coreset has a total of 50 sequences (625 total passes).

\textbf{Top-45} We also consider how well we would do if the oracle is only allowed to use the most popular pass sequences but limited to 45 passes. We use 45 passes because this is maximum allowed for our all other baselines and our proposed method.

\textbf{RL-PPO} We reproduce the Autophase~\citep{AutoPhase} pipeline by using the state-of-the-art RL algorithm PPO~\citep{ppo} to learn a policy model. We have two program representations for training the RL models, including the Autophase feature and the \programl graphs (note that \citet{AutoPhase} only used the Autophase feature). The Autophase/\programl feature is sent to a GNN/MLP for feature encoding, which outputs a program-level embedding. 
Following~\citet{AutoPhase}, we add an additional action history feature to the RL pipeline, which is a histogram of previously applied passes. The vector of the histogram of action history is divided by 45 (i.e., the number of the total passes in our budget) for normalization. A 2-layer MLP is used to encode the action history to obtain a feature vector, which is concatenated with the program embedding extracted from the \programl graph or the Autophase feature. The concatenated feature is sent to a 2-layer MLP to output the action probability for the policy network. The value network (i.e., the critic) in our PPO pipeline mimics the policy network (i.e., the actor) in feature encoding and outputs a scalar to estimate the state values. The state values are the expectation of the discounted cumulative rewards where the reward in each step is the improvement over \texttt{-O0}: $(I^{(t)}_p - I^{\pi_\theta}_p)/{I^{(t)}_p}$, where ${I^{(t)}_p}$ denotes the current IR instruction count of the program $p$ at time step $t$. This reward is reasonable since it makes the value approximation Markovian. At inference, an action is sampled from the output of the learned policy network at each time step until the total number of steps reaches 45. 

\textbf{Q-value-rank} We consider each pass sequence in the coreset as a \emph{generalized action} and train a Q network to predict the value of each generalized action. Recall that the value vector $\vr^p$ is the highest cumulative reward observed during the rollout of each pass sequence in the coreset on program $p$. The Q-value-rank model is trained to approximate the value vector using a mean squared loss. 

\textbf{BC} We consider learning a standard behavior cloning model to predict the best pass sequences from the coreset, where the best pass sequence is defined as the following. As in the previous bullet point, the value vector is denoted by $\vr^p$. If there is only one $i$ such that $r_i^p = \max_{j \in n} r_j^p$, then the classification label is $i$. If there are multiple such $i$'s (multiple pass sequences) that achieve the largest reward $\max_{j \in n} r_j^p$, then we order the corresponding pass sequences by length-lexicographic ordering. The classification label is selected to be the first one after the ordering. This ensures that our definition for the best pass sequence (among the coreset) for a program is unique and consistent. We use a standard cross entropy loss to train the single-class classification model.

\textbf{NVP} This is the normalized value prediction method described in Section~\ref{sec:nvp}.

The last three methods (i.e., Q-value-rank, BC, and NVP) share the same \textbf{inference protocol}. Note that they all output a vector $\va$ of length 50 whose entries correspond to the pass sequences in the coreset. At inference, we roll out the pass sequences with the highest values in $\va$ one by one until our budget of 45 passes is reached. Since the pass sequence has an average length of 12.5, typically 3 or 4 pass sequences are applied (anything beyond 45 passes will be truncated). For BC and NVP, we also tried sampling pass sequences using the model output $\va$, but that resulted in worse performance. Therefore, we stick to selection by maximum values.


\subsection{Main Results}

\begin{table}[t!]
\centering
\begin{tabular}{@{}lrrr@{}}
\toprule

Method                   & \#passes &\quad $\Bar{I}^{Oz}$ & \quad $\Bar{I}_{G}^{Oz}$ \\ \midrule
Compiler (\texttt{-Oz})  & -        & 0\%                 & 1.000 \\ \midrule
\texttt{Autophase-RL-PPO}    & 45       & -16.3\%             & 0.960            \\
\texttt{GCN-RL-PPO}          & 45       & -12.2\%             & 0.987            \\
\texttt{GGC-RL-PPO}          & 45       & \textbf{-8.5}\% & \textbf{1.000}            \\
\texttt{GIN-RL-PPO}          & 45       & -11.3\%         & 0.991            \\
\texttt{GAT-RL-PPO}          & 45       & -10.3\%         & 0.999            \\
\texttt{GEAN-RL-PPO}          & 45       & -10.0\%         & 0.997       \\ \midrule
\texttt{GEAN-Q-value-rank}     & 45       & \textbf{-0.3\%} & \textbf{1.043}       \\ \midrule
\texttt{Autophase-BC}   & 45       & \textbf{2.6\%}  & \textbf{1.043}     \\
\texttt{GEAN-BC}         & 45       & 2.1\%           & 1.038          \\ \midrule
\texttt{Autophase-NVP} (Ours)   & 45       & 3.8\%           & 1.054                \\
\texttt{GCN-NVP} (Ours)         & 45       & 4.4\%           & 1.058                \\
\texttt{GGC-NVP} (Ours)         & 45       & 4.1\%           & 1.056                \\
\texttt{GIN-NVP} (Ours)         & 45       & 4.2\%           & 1.056                \\
\texttt{GAT-NVP} (Ours)         & 45       & 4.1\%           & 1.055                \\
\texttt{GEAN-NVP} (Ours)   & 45       & \textbf{4.7\%}  & \textbf{1.062}   \\  \midrule

Top-45       & 45       & -7.5\%            & 0.992             \\ 
Oracle        & 625       & \textbf{5.8}\%            & \textbf{1.075}             \\

\bottomrule
\end{tabular}
\caption{Evaluation results on \textbf{held-out test set} averaged over all datasets. $\Bar{I}^{Oz}$ denotes per-program \textbf{MeanOverOz}, and $\Bar{I}_{G}^{Oz}$ denotes \textbf{GMeanOverOz} over all programs. All methods except Compiler and Oracle baselines use 45 compiler optimization passes.}
\label{table:evalSummary}
\end{table}

In Table ~\ref{table:evalSummary} we present the main results of our experiments comparing our proposed method \texttt{-NVP} to various baselines. The test programs were completely held-out during both data-driven learning phases (pass sequence search and model training).

The results show that our model achieves strong performance over the prior method (\texttt{Autophase-RL-PPO}) proposed in ~\cite{AutoPhase}. Additionally, we can see that both the GEAN model and the normalized value prediction over the discovered coreset are needed to achieve the best performance within 45 passes. See Figure~\ref{fig:trajectories} in the Appendix for a visualization of the improvement in program size over the 45 passes on programs from the holdout set.

The \texttt{Oracle} shows strong performance but requires a large number of interactions with the compiler. But, this shows that the pass sequence search generalizes to new unseen programs. This is somewhat unsurprising given that the compiler's built-in hand-tuned pass list (\texttt{-Oz}) works reasonably well for most programs.

The performance of \texttt{Top-45} by itself is weak showing that in order to achieve good results in a reasonable number of passes (45) we need to leverage a general policy and search to select the most likely candidate pass sequences to evaluate.


\subsection{Why Did the RL-PPO Baseline Fail?}
We provide an empirical analysis of why the \texttt{RL-PPO} approaches obtain much lower performance compared to our \texttt{NVP} approaches. We \emph{hypothesize} two possible reasons for the failures of \texttt{RL-PPO}. \textbf{1) Inaccurate state-value estimation results in a high variance in training.} In the PPO algorithm, we have a policy network (the actor) to output a probability distribution over the actions. And we have a value network (the critic) for estimating the state values, where the approximation is based on regressing the cumulative reward of trajectories sampled from the current policy~\citep{ppo}. The update of the policy network is based on the state value outputted by the value network. 
Inaccurate state value estimation results in a high variance in training the policy network. Due to the stochastic nature of the value estimation that stems from the Monte Carlo sampling in cumulative reward regression, it is difficult to analyze how accurately the value network approximates the ground truth state values (which are unknown even for the programs in the training set). We alleviate this issue by analyzing the \texttt{Q-value-rank} approach (as introduced in Section~\ref{sec:baseline_methods}), which can be seen as a simplified version of the value approximation in PPO. The \texttt{Q-value-rank} approach is simpler because the values to estimate are deterministic (i.e., the value vector $\vr^p$ is fixed for a program $p$). Moreover, since we consider the 50 pass sequences in our coreset as 50 generalized actions, the \texttt{Q-value-rank} approach can be seen as the value approximation in a PPO pipeline where a trajectory consists of only a single step over the 50 generalized actions. In this sense, the \texttt{Q-value-rank} approach is a simplified version of the regular value estimation in PPO. Figure~\ref{fig:gnn_qvalue} shows that the value estimation is inaccurate for programs in the held-out test set even for \texttt{Q-value-rank} approach. Therefore, it is even more challenging to estimate the state values in PPO. The inaccuracy leads to a high variance in training the policy network. 
\textbf{2) The reward is very sparse.} As shown in Figure~\ref{fig:reward_matrix}, the rewards are very sparse. Therefore, the good states (i.e., the program states with a higher chance to be optimized in code size reduction) are rarely seen by the value/policy network during rollouts. Then, the value network does not have a good value estimation for those good states, and the policy network does not converge to output a good policy for them. 
We conjecture these two issues are the main reason for why the \texttt{RL-PPO} methods obtain the worst performance as shown in Table~\ref{table:evalSummary}.

\subsection{Ablation Studies} \label{ablation-studies}

\textbf{Ablation for GEAN-NVP} We perform 3 ablation experiments for \texttt{GEAN-NVP}, where we remove graph mixup, mask the edge embedding, and remove the type graph, respectively. The results in Table~\ref{tab:gnn_ablation} show that the test MeanOverOz metric drops after removing any of the three components. Specifically, the performance drops significantly after removing the type graph, which validates its importance.

\textbf{The effect of the temperature} The temperature parameter $T$ in Eq.~\ref{eq:normalizing_values} controls how sharp the target distribution is. The distribution tends to be sharper as the temperature decreases. To analyze the influence of the temperature on the generalization of the model, We vary the temperature $T$ in training the GEAN-NVP model and report the results in Figure~\ref{fig:temp_ablation}.

\begin{figure}[t!]
    \centering
    \includegraphics[width=0.7\linewidth]{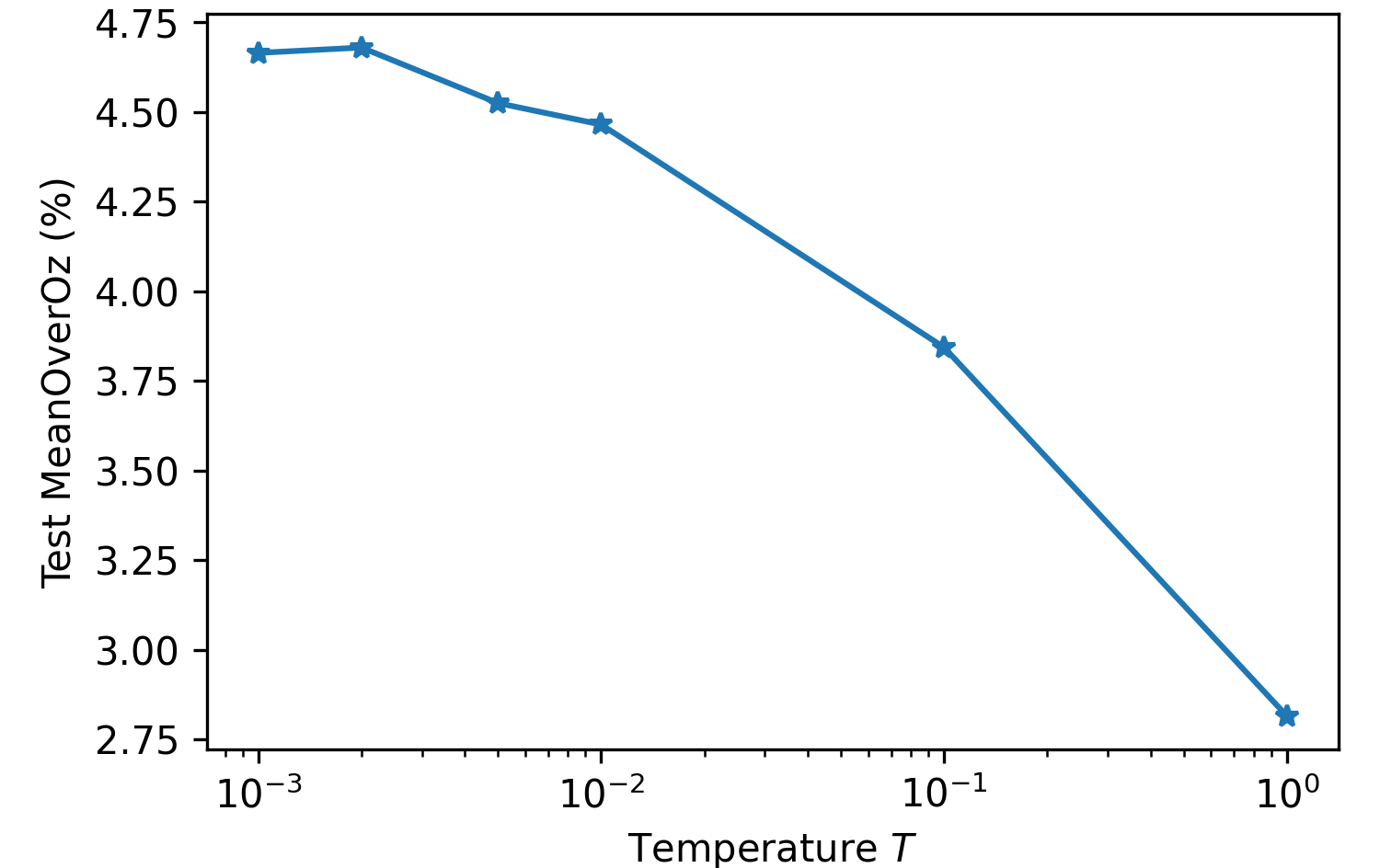}
    \caption{The effect of temperature on GEAN-NVP.}
    \label{fig:temp_ablation}
\end{figure}

\begin{table}[t!]
    \centering
    \resizebox{.65\columnwidth}{!}{%
    \begin{tabular}{lr}
    \toprule
    \textbf{Method} & \textbf{Test MeanOverOz} \\
    \midrule
    \texttt{GEAN-NVP} & 4.7\% (0.0\%) \\
    \midrule
    - graph mixup & 4.4\% (-0.3\%) \\
    - edge embedding & 4.4\% (-0.3\%)\\
    - type graph & -5.3\% (-10.0\,\%)\\
    \bottomrule
    \end{tabular}%
    }
    \caption{Ablation on GEAN-NVP components.}
    \label{tab:gnn_ablation}
\end{table}

\section{Conclusions}
In this paper, we develop a pipeline for program size reduction under limited compilation passes. We find that it is a great challenge to approximate the state values (i.e., the maximum code size reduction) for a diverse set of programs, so existing state-of-the-art methods such as proximal policy optimization (PPO) fail to obtain good performances. To tackle this problem, we propose a search algorithm that discovers a good set of pass sequences (i.e., the coreset), which generalizes well to unseen programs. Moreover, we propose to train a GNN to approximate the normalized state values of programs over the coreset, for which we propose a variant of the graph attention network, termed GEAN. Our pipeline of coreset discovery and normalized value prediction via GEAN perform significantly better than the PPO baselines. 
\pagebreak

\bibliography{main}
\bibliographystyle{icml2023}

\clearpage
\onecolumn

\appendix
\section{Proofs}
\begin{lemma}
The objective $J(S)$ defined in Eqn.~\ref{eq:obj-picking-coreset} 
\begin{equation}
\max_{|S| \le K} J(S) = \sum_{i=1}^N \max_{j \in S} r_{ij}
\end{equation}
(with the additional definition $J(\emptyset) = 0$), is a nonnegative and monotone submodular function. 
\end{lemma}
\begin{proof}
Since $r_{ij} > 0$, it is clear that $J(S) \ge 0$ is nonnegative. 

To incorporate the special case $J(\emptyset) = 0$, note that $J(S)$ can be written as
\begin{equation}
\max_{|S| \le K} J(S) = \sum_{i=1}^N \max\left(\max_{j \in S} r_{ij}, 0\right).
\end{equation}
Let $\hat r_{ij} = r_{ij}$ and $\hat r_{i,0} = 0$, then in order to prove $J(S)$ is monotone and submodular, by additivity, we only need to prove $J_i(S) := \max_{j\in S \cup \{0\}} \hat r_{ij}$ is monotone and submodular. 

For any $A\subseteq B$, it is clear that 
\begin{equation}
   J_i(A) = \max_{j\in A \cup \{0\}} \hat r_{ij} \le \max_{j\in B \cup \{0\}} \hat r_{ij} = J_i(B)
\end{equation}
So $J_i(S)$ is monotone. 

To prove submodularity, for any $A\subseteq B$, we comare the quatity of $J_i(A\cup \{j\} ) - J_i(A)$ and $J_i(B\cup \{j\} ) - J_i(B)$ for $j \notin B$.

\textbf{Case 1: $\hat r_{ij}$ is a maximum over the subset $B$}. In this case, then $\hat r_{ij}$ is also a maximum over the subset $A$. Then $J_i(A\cup \{j\}) = J_i(B\cup\{j\}) = \hat r_{ij}$, since $J_i(A) \le J_i(B)$, we have:
\begin{equation}
    J_i(A\cup \{j\} ) - J_i(A) \ge J_i(B\cup \{j\} ) - J_i(B) \label{eq:submodularity-Ji}
\end{equation}

\textbf{Case 2: $\hat r_{ij}$ is a maximum over $A$ but not in $B$}. Then $J_i(A\cup \{j\}) - J_i(A) \ge 0$, but $J_i(B\cup \{j\}) - J_i(B) = 0$. So Eqn.~\ref{eq:submodularity-Ji} still holds. 

\textbf{Case 3: $\hat r_{ij}$ is neither a maximum in $A$ or in $B$}. Then both $J_i(A\cup \{j\}) - J_i(A) = 0$ and $J_i(B\cup \{j\}) - J_i(B) = 0$. So Eqn.~\ref{eq:submodularity-Ji} still holds.

By definition of submodularity (Eqn.~\ref{eq:submodularity-Ji}), we know $J_i(S)$ is submodular and so does $J(S)$.
\end{proof}

\section{GEAN Encoding}

Our Graph Edge Attention Network (GEAN) has the following key features.

\textbf{Attention with edge features} We modify the attention mechanism in GAT to output an edge embedding and two node features for neighborhood aggregation. 
For clarity, we show a table containing the notations used in the GNN in Table~\ref{tab:gnn_notations}. 
The feature update process can be mathematically defined by the following equations, where $M_i, i = 1,\dots, 5$ is an encoding fully connected layer.

\begin{figure}
    \centering
    \includegraphics[width=.5\textwidth]{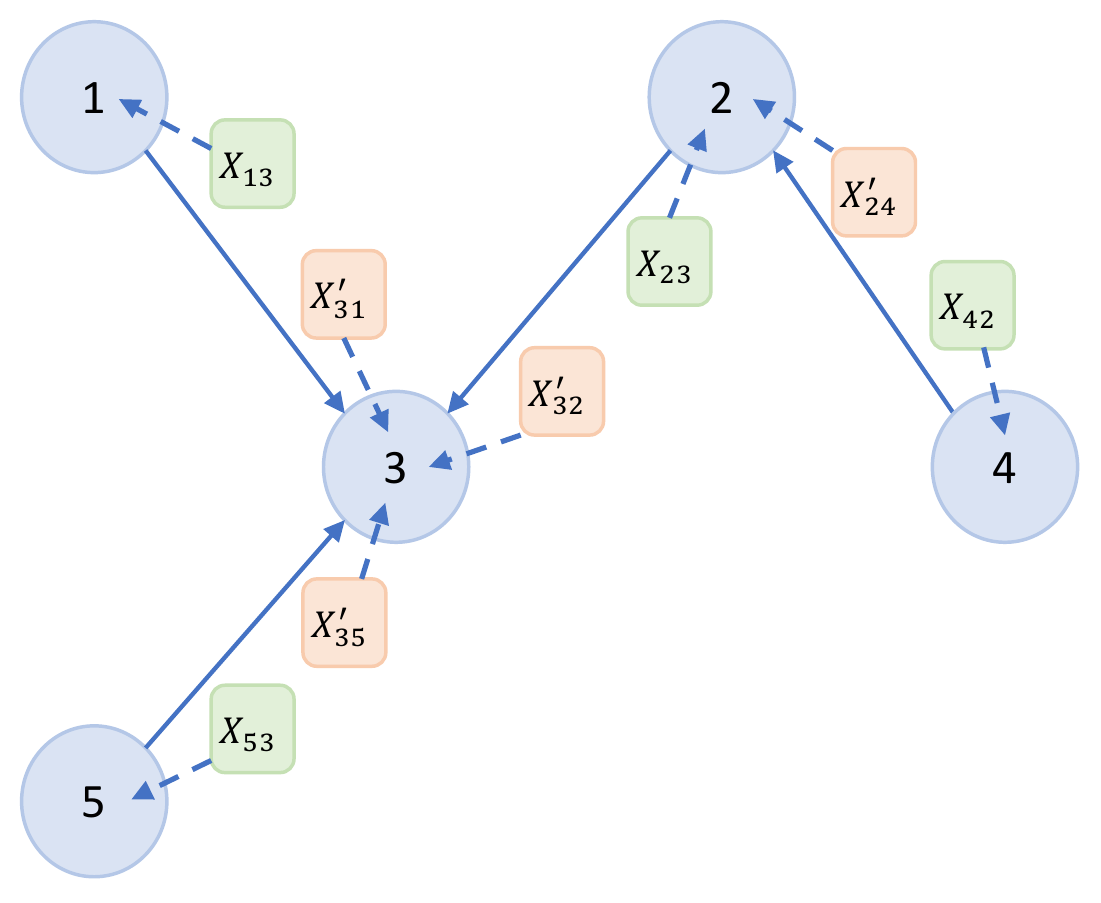}

    \caption{Graph attention. Circles denote nodes and solid arrows denote edges. Squares are the calculated features, and dash arrows represent feature aggregation. The orange/green squares denote the features to be aggregated in the target/source nodes of the edges. The edge embedding and the attention are not shown.}
   
    \label{graph_attention}
\end{figure}

\begin{table}[t!]
\centering
\begin{tabular}{ll}
\toprule
Notation & Meaning \\
\midrule
$\gE$ & The set of edges in the graph \\
$(i,j)$ & Edge from node $i$ pointing to node $j$ \\
$X_i^{(t)}$ &Representation of node $i$ at layer $t$ \\
$E_{i \to j}^{(t)}$ &Representation of the edge $(i,j)$ at layer $t$ \\
$X_{ij}^{(t)}$ &Representation for node $i$ associated with edge $(i, j)$ \\
$X_{ij}^{\prime(t)}$ &Representation for node $i$ associated with edge $(j, i)$ \\
$a_{ij}^{(t)}$ & Raw attention associated with representation $X_{ij}^{(t)}$ \\
$a_{ij}^{\prime(t)}$ &  Raw attention associated with representation $X_{ij}^{\prime(t)}$ \\
$\alpha_{ij}^{(t)}$ & Normalized attention associated with $a_{ij}^{(t)}$ \\
$\alpha_{ij}^{\prime(t)}$ &  Normalized attention associated with $a_{ij}^{\prime(t)}$ \\
$\mathcal{T}_i$ & Target neighbors of node $i$: $\{j | (i,j) \in \gE\}$ \\
$\mathcal{S}_i$ & Source neighbors of node $i$: $\{j | (j,i) \in \gE\}$ \\
\bottomrule
\end{tabular}
\caption{The notations in GEAN.}
\label{tab:gnn_notations}
\end{table}

\begin{align}
&X_{ij}^{\prime(t+1)} = M_1(X_i^{(t)}, E_{i \to j}^{(t)}, X_j^{(t)}), \label{eq:target_feat}\\
&a_{ij}^{\prime(t+1)} = M_2(X_i^{(t)}, E_{i \to j}^{(t)}, X_j^{(t)}), \label{eq:target_attn}\\
&X_{ji}^{(t+1)} = M_3(X_i^{(t)}, E_{i \to j}^{(t)}, X_j^{(t)}), \label{eq:source_feat}\\
&a_{ji}^{(t+1)} = M_4(X_i^{(t)}, E_{i \to j}^{(t)}, X_j^{(t)}), \label{eq:source_attn}\\
&E_{i \to j}^{(t+1)} = M_5(X_i^{(t)}, E_{i \to j}^{(t)}, X_j^{(t)}), \label{eq:edge_feat}
\end{align}

In words, the node-edge-node triplet, $(X_i^{(t)}, E_{i \to j}^{(t)}, X_j^{(t)})$, associated with edge $(i, j)$, is encoded by fully connected layers to output 5 features, including $X_{ij}^{\prime(t+1)}$ and $a_{ij}^{\prime(t+1)}$ (a representation and attention to be aggregated in node $i$), and $X_{ji}^{(t+1)}$ and $a_{ji}^{(t+1)}$ (a representation and attention to be aggregated in node $j$), and the updated edge representation $E_{i \to j}^{(t+1)}$. Note that the features to be aggregated to a target node are marked with the $^\prime$, and those to a source node are without the $^\prime$ (see Figure~\ref{graph_attention}).
After the feature encoding, we perform an attention-weighted neighborhood aggregation for each node, which can be mathematically described by the following equations. 
\begin{equation}
    \left[\left[\alpha_{ij}^{(t+1)}\right]_{j \in \mathcal{T}_i} \left\| \left[\alpha_{ij}^{\prime(t+1)}\right]_{j \in \mathcal{S}_i} \right. \right] 
    = \operatorname{Softmax}\left[\left[a_{ij}^{(t+1)}\right]_{j \in \mathcal{T}_i} \left\| \left[a_{ij}^{\prime(t+1)}\right]_{j \in \mathcal{S}_i} \right. \right]
\end{equation}
\begin{align}
X_{i}^{(t+1)} = \sum_{j \in \mathcal{T}_i} \alpha_{ij}^{(t+1)} X_{ij}^{(t+1)} + \sum_{j \in \mathcal{S}_i} \alpha_{ij}^{\prime(t+1)} X_{ij}^{\prime(t+1)}
\end{align}
where $\|$ denotes concatenation.

In comparison, GAT only outputs an attention score by encoding the node-edge-node triplet:  $a_{ij}^{(t+1)} = (X_i^{(t)}, E_{i \to j}^{(t)}, X_j^{(t)})$, and the feature for neighborhood aggregation is only conditioned on the neighbors: $X_{ij}^{(t+1)} = \mathrm{MLP}(X_j^{(t)})$. To summarize, our encoding approach can ensure that the GNN model is aware of the direction of the edge and that the edge embedding is updated in each layer, which helps improve the performance (as shown in Table~\ref{table:evalSummary}).

\section{Detailed Results}

\begin{table}[h!]

\centering
\begin{tabular}{@{}lrrrrrr@{}}
\toprule

Dataset  &   Oracle  &  Top-45  &  Autophase-RL-PPO  &  Autophase-NVP  &  GEAN-RL-PPO  &  GEAN-NVP  \\ \midrule

anghabench-v1 & 0.7\%/1.011 & -1.0\%/0.996 & -15.9\%/0.974 & -0.2\%/1.002 & -0.8\%/0.996 & -0.0\%/1.003  \\
blas-v0 & 2.6\%/1.028 & -0.4\%/0.997 & -1.7\%/0.984 & 2.1\%/1.023 & -1.0\%/0.990 & 2.4\%/1.026  \\
cbench-v1 & 3.5\%/1.041 & -2.4\%/0.984 & -10.1\%/0.925 & -0.1\%/1.008 & -1.6\%/0.998 & 2.2\%/1.028  \\
chstone-v0 & 9.3\%/1.106 & 1.2\%/1.016 & 1.3\%/1.018 & 8.3\%/1.095 & 5.4\%/1.060 & 8.8\%/1.101  \\
clgen-v0 & 5.4\%/1.060 & 3.1\%/1.034 & -0.5\%/0.998 & 4.6\%/1.051 & 0.3\%/1.005 & 5.0\%/1.056  \\
csmith-v0 & 21.2\%/1.320 & -96.3\%/0.851 & -116.0\%/0.954 & 21.1\%/1.320 & -124.6\%/0.965 & 21.1\%/1.320  \\
github-v0 & 1.0\%/1.011 & 0.2\%/1.002 & 0.1\%/1.001 & 0.9\%/1.010 & -0.2\%/0.999 & 0.9\%/1.010  \\
linux-v0 & 0.6\%/1.007 & -0.4\%/0.998 & -0.5\%/0.997 & 0.6\%/1.006 & -2.3\%/0.989 & 0.6\%/1.007  \\
llvm-stress-v0 & 6.3\%/1.087 & -18.9\%/0.885 & -67.0\%/0.731 & 0.7\%/1.035 & -17.5\%/0.888 & 2.1\%/1.045  \\
mibench-v1 & 1.7\%/1.020 & 0.0\%/1.003 & -2.8\%/0.976 & -5.8\%/0.963 & -0.3\%/1.000 & -0.1\%/1.003  \\
npb-v0 & 9.8\%/1.159 & 5.7\%/1.085 & 0.9\%/1.035 & 5.1\%/1.079 & 3.7\%/1.068 & 5.5\%/1.085  \\
opencv-v0 & 5.2\%/1.061 & 1.0\%/1.013 & 0.5\%/1.007 & 4.2\%/1.051 & 0.3\%/1.004 & 4.8\%/1.057  \\
poj104-v1 & 7.8\%/1.105 & 3.9\%/1.055 & -17.5\%/0.876 & 6.1\%/1.080 & -0.7\%/1.008 & 6.3\%/1.082  \\
tensorflow-v0 & 6.1\%/1.077 & -0.2\%/0.998 & 0.2\%/1.004 & 5.9\%/1.075 & -0.2\%/0.998 & 5.9\%/1.075  \\

\midrule

Average & 5.8\%/1.075 & -7.5\%/0.992 & -16.3\%/0.960 & 3.8\%/1.054 & -10.0\%/0.997 & 4.7\%/1.062  \\

\bottomrule
\end{tabular}
\caption{Detailed evaluation results on held-out test sets.}
\end{table}

\clearpage

\section{Trajectory comparison}
\begin{figure}[H]
    \centering
    \includegraphics[width=.85\textwidth]{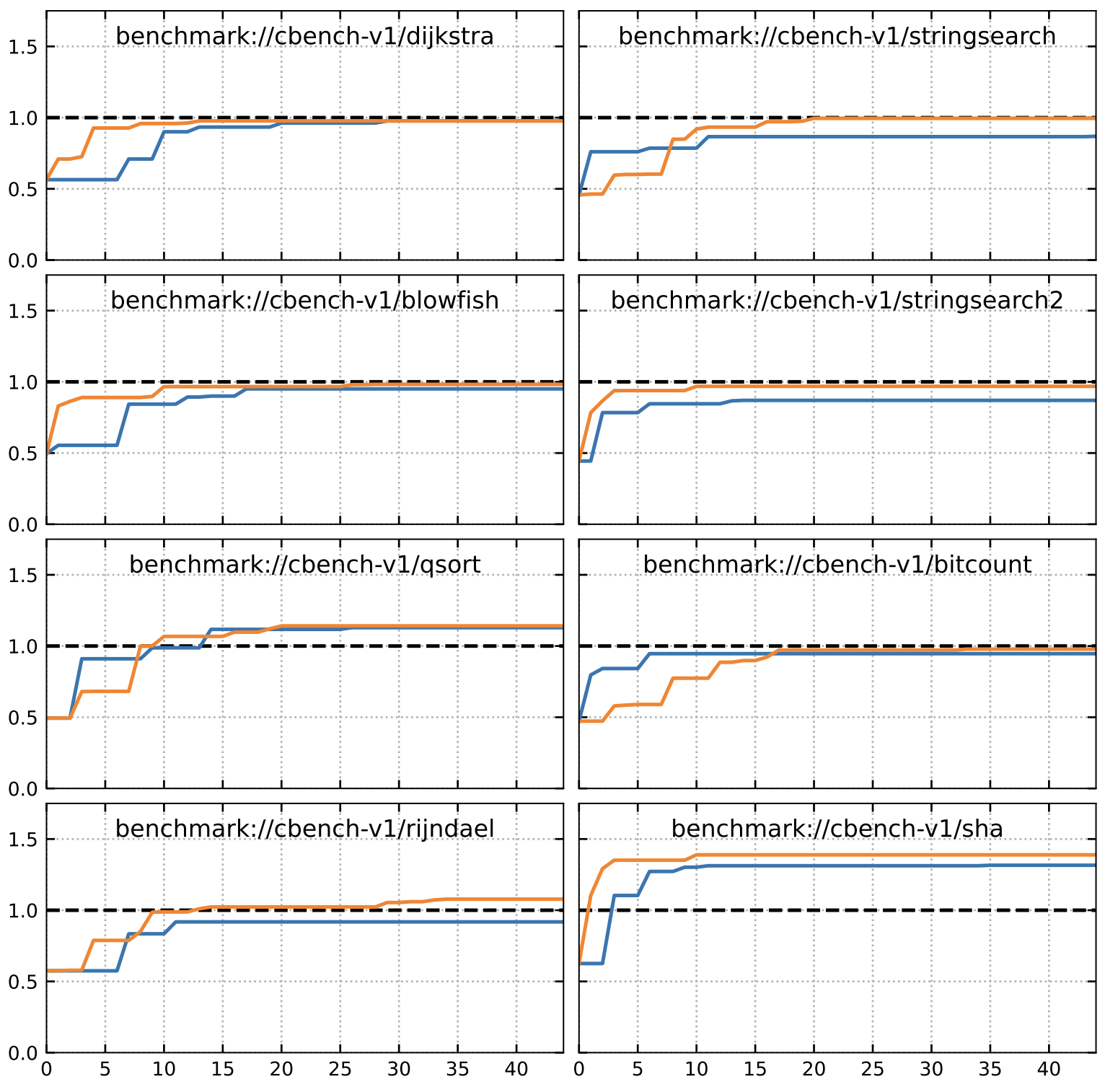}
    \caption{Program optimization example over many steps comparing the \textbf{Autophase-RL-PPO} (blue) approach with our \textbf{GEAN-NVP} (orange) approach. The dashed line represents the compiler default \texttt{-Oz} performance and higher is better.}
    \label{fig:trajectories}
\end{figure}
\clearpage

\section{Compiler passes}

\begin{table}[H]
	\begin{center}
		\begin{tabular}{rlrlrl}
			\toprule
			Index & Flag & Index & Flag & Index & Flag \\
			\midrule
			0 & -add-discriminators & 42 & -globalsplit & 84 & -lower-expect \\
			1 & -adce & 43 & -guard-widening & 85 & -lower-guard-intrinsic \\
			2 & -aggressive-instcombine & 44 & -hotcoldsplit & 86 & -lowerinvoke \\
			3 & -alignment-from-assumptions & 45 & -ipconstprop & 87 & -lower-matrix-intrinsics \\
			4 & -always-inline & 46 & -ipsccp & 88 & -lowerswitch \\
			5 & -argpromotion & 47 & -indvars & 89 & -lower-widenable-condition \\
			6 & -attributor & 48 & -irce & 90 & -memcpyopt \\
			7 & -barrier & 49 & -infer-address-spaces & 91 & -mergefunc \\
			8 & -bdce & 50 & -inferattrs & 92 & -mergeicmps \\
			9 & -break-crit-edges & 51 & -inject-tli-mappings & 93 & -mldst-motion \\
			10 & -simplifycfg & 52 & -instsimplify & 94 & -sancov \\
			11 & -callsite-splitting & 53 & -instcombine & 95 & -name-anon-globals \\
			12 & -called-value-propagation & 54 & -instnamer & 96 & -nary-reassociate \\
			13 & -canonicalize-aliases & 55 & -jump-threading & 97 & -newgvn \\
			14 & -consthoist & 56 & -lcssa & 98 & -pgo-memop-opt \\
			15 & -constmerge & 57 & -licm & 99 & -partial-inliner \\
			16 & -constprop & 58 & -libcalls-shrinkwrap & 100 & -partially-inline-libcalls \\
			17 & -coro-cleanup & 59 & -load-store-vectorizer & 101 & -post-inline-ee-instrument \\
			18 & -coro-early & 60 & -loop-data-prefetch & 102 & -functionattrs \\
			19 & -coro-elide & 61 & -loop-deletion & 103 & -mem2reg \\
			20 & -coro-split & 62 & -loop-distribute & 104 & -prune-eh \\
			21 & -correlated-propagation & 63 & -loop-fusion & 105 & -reassociate \\
			22 & -cross-dso-cfi & 64 & -loop-guard-widening & 106 & -redundant-dbg-inst-elim \\
			23 & -deadargelim & 65 & -loop-idiom & 107 & -rpo-functionattrs \\
			24 & -dce & 66 & -loop-instsimplify & 108 & -rewrite-statepoints-for-gc \\
			25 & -die & 67 & -loop-interchange & 109 & -sccp \\
			26 & -dse & 68 & -loop-load-elim & 110 & -slp-vectorizer \\
			27 & -reg2mem & 69 & -loop-predication & 111 & -sroa \\
			28 & -div-rem-pairs & 70 & -loop-reroll & 112 & -scalarizer \\
			29 & -early-cse-memssa & 71 & -loop-rotate & 113 & -separate-const-offset-from-gep \\
			30 & -early-cse & 72 & -loop-simplifycfg & 114 & -simple-loop-unswitch \\
			31 & -elim-avail-extern & 73 & -loop-simplify & 115 & -sink \\
			32 & -ee-instrument & 74 & -loop-sink & 116 & -speculative-execution \\
			33 & -flattencfg & 75 & -loop-reduce & 117 & -slsr \\
			34 & -float2int & 76 & -loop-unroll-and-jam & 118 & -strip-dead-prototypes \\
			35 & -forceattrs & 77 & -loop-unroll & 119 & -strip-debug-declare \\
			36 & -inline & 78 & -loop-unswitch & 120 & -strip-nondebug \\
			37 & -insert-gcov-profiling & 79 & -loop-vectorize & 121 & -strip \\
			38 & -gvn-hoist & 80 & -loop-versioning-licm & 122 & -tailcallelim \\
			39 & -gvn & 81 & -loop-versioning & 123 & -mergereturn \\
			40 & -globaldce & 82 & -loweratomic &  &  \\
			41 & -globalopt & 83 & -lower-constant-intrinsics &  &  \\
			\bottomrule
		\end{tabular}
	\end{center}
	\caption{A list of LLVM compiler pass indices and their corresponding command line flag.}
	\label{table:llvm_pass_table}
\end{table}

\end{document}